\documentclass[letterpaper,journal]{IEEEtran}
\usepackage{amsmath,amsfonts,amssymb}
\usepackage{algorithmic}
\usepackage{algorithm}
\usepackage{array}
\usepackage{booktabs}
\usepackage{textcomp} 
\usepackage{stfloats}
\usepackage{url}
\usepackage{graphicx}
\usepackage{cite}
\usepackage{placeins}
\graphicspath{{figures/}{Fig/}}
\newtheorem{proposition}{Proposition}

\begin{document}

\title{Dependency-Aware HARQ and Link Adaptation for Wireless Transmission of Open-Vocabulary Scene Graphs}

\author{Yuli~Liu, Jiacheng~Ruan, and Caiming~Sun,~\IEEEmembership{Senior Member,~IEEE}%
\thanks{This work was supported in part by the National Natural Science Foundation of China under Grant 62175120 and in part by the Shenzhen Science and Technology Innovation Program under Grant JCYJ20220818103011023.}%
\thanks{This work has been submitted to the IEEE for possible publication. Copyright may be transferred without notice, after which this version may no longer be accessible.}%
\thanks{Y. Liu is with the Department of Ocean Science and Technology, University of Macau, Taipa, Macau, China (e-mail: mc45259@um.edu.mo).}%
\thanks{J. Ruan and C. Sun are with the Shenzhen Institute of Artificial Intelligence and Robotics for Society (AIRS), School of Science and Engineering, The Chinese University of Hong Kong (CUHK), Shenzhen 518172, Guangdong, China (e-mail: ruanjiacheng@cuhk.edu.cn; cmsun@cuhk.edu.cn). C. Sun is the corresponding author.}%
}

\maketitle

\begin{abstract}
Wireless visual uplinks increasingly carry structured representations for edge inference, making packet reliability part of task-aware link adaptation. In open-vocabulary scene-graph transmission, an indexed triplet is usable only if both its triplet packet and the vocabulary packets defining any newly introduced tokens are recovered. This prerequisite coupling makes the marginal value of packet reliability depend on neighboring packet reliabilities. We formulate a dependency-aware semantic distortion and jointly optimize finite choices of modulation and coding scheme (MCS), transmit power, and Chase-combining hybrid automatic repeat request (HARQ) depth under expected delay and energy constraints. The distortion is multi-affine in packet failure probabilities and cannot, in general, be reduced to static separable unequal error protection (UEP) weights when prerequisites are active. This structure yields a state-dependent reliability coefficient and explicit switching thresholds among wireless actions. A Lagrangian block method performs exact per-packet finite-action updates for fixed multipliers. On reduced instances, it matches exhaustive optimization in 28 of 30 cases, with a worst gap of 1.095\%. On GQA traces using a table-driven block error rate (BLER) abstraction, it reduces mean semantic distortion by 58.71\% and grounded-query failure by 57.03\% relative to dependency-agnostic HARQ under the same budgets.
\end{abstract}


\section{Introduction}
\label{sec:introduction}

\IEEEPARstart{W}{ireless} visual uplinks enable cameras and mobile devices to deliver compact scene information to edge processors, where packet reliability must be managed under limited latency and energy. Conventional wireless reliability control uses rate adaptation, transmit-power control, hybrid automatic repeat request (HARQ), and unequal error protection (UEP) to trade transmission resources for lower packet error rates \cite{ref5,ref6,ref9,ref10,ref11,ref13,ref14}. Such designs are naturally packet oriented: the protection assigned to a packet is determined by its payload, target reliability, or a prescribed importance measure. This treatment becomes insufficient when receiver-side usefulness has explicit prerequisites, because successful recovery of one packet can determine whether separately transmitted packets are interpretable at all. In that case, the semantic loss associated with a packet error depends not only on the failed packet itself but also on the recovery states of related packets, which can change how limited delay and energy should be allocated across the wireless link.

Task-oriented communication provides a natural setting for this dependency because the receiver need not reconstruct every source bit as long as the recovered representation preserves the information required by the downstream task \cite{ref20,ref21,ref22}. Deep joint source--channel coding (JSCC), semantic codebooks, and related task-oriented transmission methods exploit this principle by adapting source representation or protection to reconstruction quality, feature importance, or task utility \cite{ref15,ref16,ref17,ref18,ref24,ref30}. Scene graphs further provide a compact structured representation of objects, relations, and attributes for retrieval and compositional reasoning \cite{ref37,ref38,ref39}, and recent wireless designs have used them for explainable or task-oriented semantic transmission \cite{ref23,ref25,ref26}. In an open-vocabulary scene graph, however, a newly introduced token must first be associated with an index known to the receiver. An indexed triplet containing that token is therefore interpretable only if both its triplet packet and the vocabulary-increment packet carrying the required token--index mapping are successfully recovered. A single vocabulary packet can consequently enable several downstream triplets, while a reliably delivered triplet packet can remain unusable when one of its prerequisite mappings is missing.

This prerequisite relation creates a reliability coupling that is not captured by conventional fixed-importance protection. UEP assigns different protection levels according to unequal source importance \cite{ref13,ref14}, while recent semantic retransmission schemes adapt HARQ according to semantic distortion, feature importance, codebook state, or task utility \cite{ref27,ref28,ref29}; related studies have also considered cooperative-perception and JSCC-based semantic HARQ \cite{ref45,ref46}. These approaches establish the relevance of semantic importance to wireless retransmission, but their protected units are generally valued without conditioning that value on the current reliabilities of multiple prerequisite and dependent packets. With vocabulary--triplet prerequisites, by contrast, the marginal benefit of improving one packet depends on the success probabilities of neighboring packets in the prerequisite graph. The resulting reliability value is therefore state dependent and cannot, in general, be represented by a reliability-independent static UEP weight. A wireless allocation for this setting must account jointly for prerequisite-induced semantic coupling and the delay--energy--reliability tradeoff created by MCS selection, transmit power, and HARQ retransmissions.

We study this problem in a single-user wireless visual uplink with finite per-packet link-adaptation actions. The transmitter separates newly introduced vocabulary mappings from indexed scene-graph triplets, and the wireless controller selects a modulation and coding scheme (MCS) label, transmit-power level, and maximum Chase-combining HARQ depth for each packet. We first model receiver-side triplet interpretability as a function of the residual failure probabilities of the indexed-triplet packet and all of its prerequisite vocabulary packets. This model yields a task-weighted semantic distortion with mixed reliability terms and leads to a finite joint MCS--power--HARQ allocation problem under expected frame-delay and frame-energy constraints. The distortion is multi-affine in the packet failure probabilities, so with all other packet reliabilities fixed, the dependence on any individual packet failure probability is exactly affine. We exploit this structure to derive a state-dependent marginal reliability coefficient and a finite-action switching condition, and then develop a Lagrangian block method whose packet-wise updates are exact for fixed multipliers. The resulting design is evaluated with analytical fading models, GQA-derived packet traces, a non-analytic packet-level block error rate (BLER) ledger, controlled action-switching experiments, and exhaustive optimization on tractable instances.

The main contributions are summarized as follows:
\begin{itemize}
\item We model open-vocabulary scene-graph transmission with explicit vocabulary--triplet prerequisites and derive a task-weighted receiver-side interpretability distortion that separates indexed-packet loss from vocabulary-induced semantic loss.

\item We characterize the reliability coupling created by active prerequisites. The resulting distortion contains mixed reliability terms and is not, in general, equivalent to a separable UEP objective with reliability-independent packet weights. The corresponding marginal reliability coefficient depends on the current reliabilities of dependent triplets and other prerequisites, yielding an explicit switching condition among finite MCS--power--HARQ actions.

\item We formulate the joint per-packet MCS, transmit-power, and HARQ-depth allocation under expected delay and energy constraints and develop a finite-action Lagrangian block method. For fixed multipliers, each packet update is solved exactly over its admissible action set, and repeated fixed-multiplier block updates terminate after finitely many strict configuration changes.

\item We validate the proposed allocation through analytical wireless sweeps, GQA trace-driven workloads, resource-budget and wireless-uncertainty tests, a controlled action-switching experiment, a table-driven BLER link abstraction, and a strong static dependency-aware (Static-DA) UEP comparator. On reduced instances, exhaustive enumeration agrees with the proposed method in 28 of 30 cases, with a worst observed relative gap of 1.095\%; under the primary GQA protocol, the proposed allocation reduces mean semantic distortion and grounded-query failure by 58.71\% and 57.03\%, respectively, relative to dependency-agnostic HARQ.
\end{itemize}

The remainder of the paper is organized as follows. Section~\ref{sec:related} reviews scene-graph semantic transmission, HARQ, and UEP. Section~\ref{sec:system} defines the open-vocabulary packetization and wireless link model. Section~\ref{sec:allocation} develops the dependency-aware distortion, discrete allocation problem, and packet-wise link-adaptation method. Section~\ref{sec:results} presents solver validation, analytical resource tradeoffs, the Static-DA audit, GQA trace evidence, controlled action switching, and link-abstraction cross-checks. Section~\ref{sec:conclusion} concludes the paper.

\begin{table}[!t]
\caption{Key System and Optimization Quantities}
\label{tab:key_quantities}
\centering
\footnotesize
\setlength{\tabcolsep}{3.0pt}
\begin{tabular}{p{0.24\linewidth}p{0.66\linewidth}}
\toprule
Quantity & Role in the paper\\
\midrule
$\mathcal{V}_{f-1},\Delta\mathcal{V}_f$ & Shared vocabulary state and newly introduced token mappings.\\
$\mathcal{Q}_i,\ell(i)$ & Vocabulary prerequisites and indexed-triplet packet associated with triplet $i$.\\
$B_u$ & Packet payload that determines channel uses for a selected rate label.\\
$m_u,P_u,K_u$ & Rate/MCS label, transmit power, and maximum HARQ depth selected for packet $u$.\\
$\epsilon_u,\bar{\tau}_u,\bar{E}_u$ & Residual packet failure probability, expected packet delay, and expected packet energy.\\
$w_i,\bar D_f$ & Triplet task weight and dependency-aware expected semantic distortion.\\
$\alpha_u$ & State-dependent marginal reliability value used by the packet-wise wireless update.\\
$\tau_{\max},E_{\max},\delta_v$ & Frame-delay budget, frame-energy budget, and optional vocabulary reliability floor.\\
\bottomrule
\end{tabular}
\end{table}

\section{Related Work}
\label{sec:related}

\subsection{Semantic Visual and Scene-Graph Communication}

Semantic and task-oriented communication replaces exact bit reconstruction by utility measures tied to the receiver task \cite{ref20,ref21,ref22}. Deep joint source--channel coding (DeepJSCC) and feedback-aided image transmission address end-to-end image delivery over noisy channels \cite{ref15,ref16,ref17}. Semantic codebooks, semantic multiple-input multiple-output (MIMO) systems, and generative semantic transmission extend this idea to task-oriented representations and adaptive transmission \cite{ref24,ref30,ref31}.

Scene graphs encode visual content as entities and relations and have been used for retrieval and compositional reasoning \cite{ref37,ref38,ref39}. Open-vocabulary scene-graph generation (OV-SGG) further studies recognition or generation of novel entity and relation concepts at the visual extraction stage \cite{ref49,ref50}. Wireless scene-graph methods use structured graphs for explainable or task-oriented semantic transmission \cite{ref23,ref25,ref26}, while probabilistic graph compression exploits semantic regularity to omit predictable graph elements \cite{ref47}. The present work is complementary to OV-SGG: semantic extraction is treated as an upstream input, and ``open-vocabulary'' refers to the communication interface in which previously unsynchronized token strings may arrive at run time and must be delivered before their indices are interpretable. The wireless contribution is therefore prerequisite-aware reliability allocation rather than open-vocabulary visual recognition.

\subsection{HARQ, UEP, and Semantic Retransmission}

HARQ and retransmission-aware resource allocation are established mechanisms for trading reliability against delay and energy \cite{ref7,ref8,ref10}. Finite-blocklength and variable-rate studies further characterize the interaction among coding rate, retransmission, and latency \cite{ref5,ref9,ref12}, while fast-HARQ optimization explicitly accounts for feedback delay and energy \cite{ref11}. UEP assigns different protection levels to source components according to their importance \cite{ref13,ref14}.

Semantic retransmission has recently moved this principle toward task-aware HARQ. Fine-grained semantic retransmission, feature-distortion-driven SemHARQ, and semantic-codebook HARQ adapt protection according to semantic information or task distortion \cite{ref27,ref28,ref29}. Cooperative-perception HARQ and JSCC-powered semantic HARQ further couple retransmission with task-oriented representations \cite{ref45,ref46}. The distinction considered here is that the value of a vocabulary packet is not fixed by its own content; it depends on the current reliabilities of multiple separately transmitted triplets and prerequisites, which creates a state-dependent packet value inside the wireless allocation problem.

\section{System Model}
\label{sec:system}

Fig.~\ref{fig:framework} shows the considered single-user wireless visual uplink, in which the transmitter forms an open-vocabulary scene graph, separates newly introduced token mappings from indexed triplets, and sends the resulting packets sequentially with acknowledgment/negative acknowledgment (ACK/NACK) feedback. Scene-graph representations and their use for semantic visual transmission motivate the structured triplet interface considered here \cite{ref25,ref37,ref38}.

\begin{figure*}[!t]
\centering
\includegraphics[width=\textwidth]{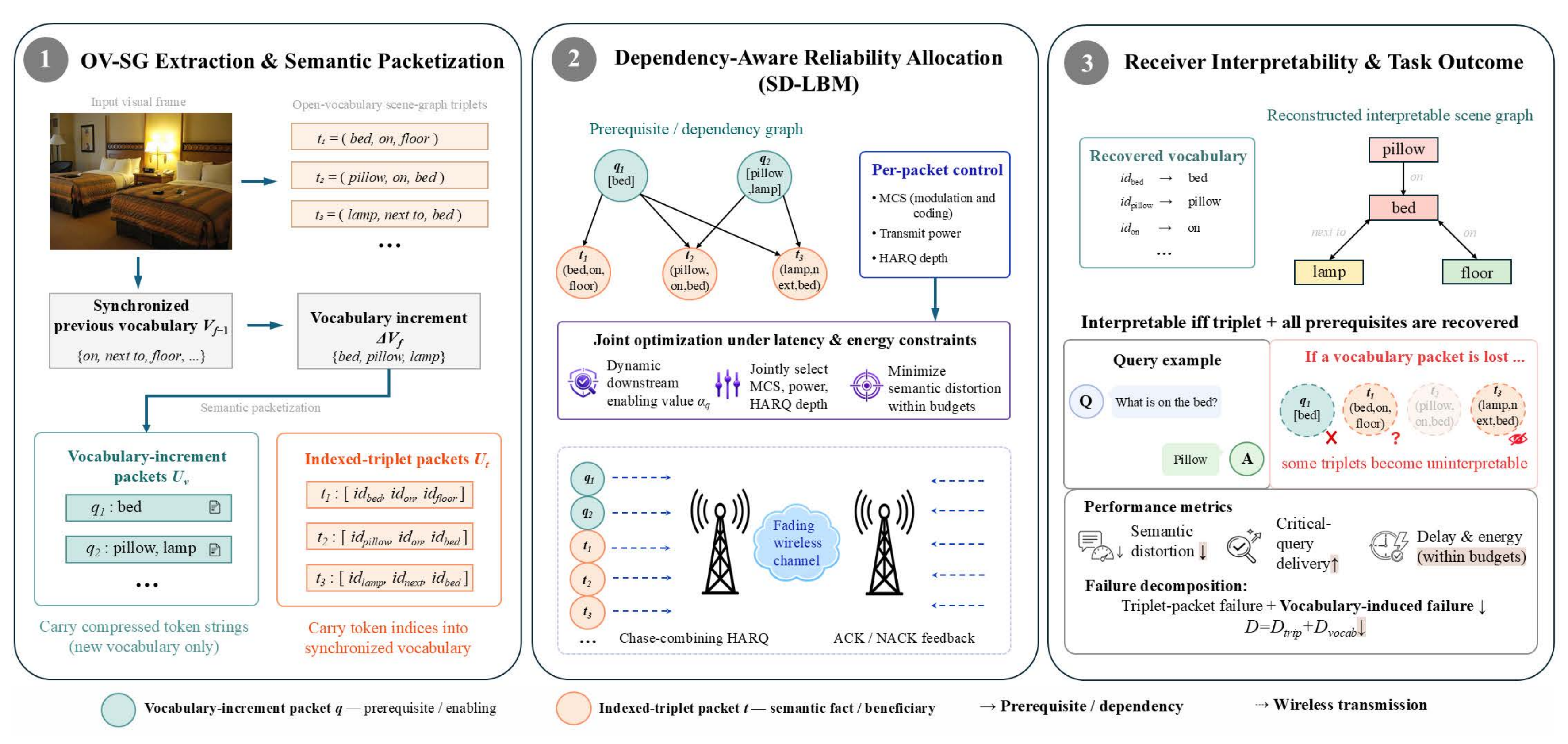}
\caption{Dependency-aware open-vocabulary scene-graph wireless semantic uplink. Vocabulary-increment packets synchronize newly used token mappings, indexed-triplet packets carry facts by index, and the wireless controller selects an MCS label, transmit power, and HARQ depth under expected delay and energy budgets. A triplet is useful only when its indexed packet and all prerequisite mappings are recovered.}
\label{fig:framework}
\end{figure*}

The wireless design begins after semantic extraction, so the realized triplets, token strings, prerequisite sets, and task weights are inputs to the packetizer and are not optimized by the wireless controller. All packet actions for a frame are selected before transmission of the first packet. Packets are then sent sequentially; ACK/NACK is used for within-packet HARQ termination and for committing successfully delivered vocabulary mappings to the state used by subsequent frames, but it does not trigger same-frame cross-packet action re-optimization. This frame-wise open-loop allocation isolates the prerequisite-induced reliability valuation studied here from causal scheduling or sequential decision control. The resulting expected frame delay and energy are additive across packets. Conditional on the frame-level large-scale state, the analytical branch uses independent Rayleigh fading blocks across packets and HARQ rounds, following the standard block-fading abstraction used in wireless reliability analysis \cite{ref2,ref4}.

\subsection{Semantic Packetization and Vocabulary State}
\label{subsec:packetization}

Frame $f$ contains
\begin{equation}
\mathcal{T}_f=\left\{t_i=(s_i,r_i,o_i,c_i)\right\}_{i=1}^{N_f},
\label{eq:triplet_set}
\end{equation}
where $s_i$, $r_i$, and $o_i$ are the subject, relation, and object tokens, $c_i\in(0,1]$ is the extractor confidence, and $N_f$ is the number of extracted triplets. The triplet representation follows the subject--relation--object structure used in scene-graph reasoning and scene-graph semantic communication \cite{ref37,ref39,ref25}.

Before frame $f$, the transmitter and receiver share vocabulary $\mathcal{V}_{f-1}$, while the current token set and vocabulary increment are
\begin{align}
\mathcal{X}_f&=\bigcup_{i=1}^{N_f}\{s_i,r_i,o_i\},\label{eq:frame_tokens}\\
\Delta\mathcal{V}_f&=\mathcal{X}_f\setminus\mathcal{V}_{f-1}.\label{eq:vocab_increment}
\end{align}
After synchronization, each token is represented by an index of length
\begin{equation}
b_{\rm id}(f)=\left\lceil\log_2\left|\mathcal{V}_{f-1}\cup\Delta\mathcal{V}_f\right|\right\rceil.
\label{eq:id_length}
\end{equation}

The triplets are sorted by fixed task weight and partitioned into $L$ indexed-triplet packets. Let $\ell(i)\in\{1,\ldots,L\}$ denote the packet containing $t_i$, let $q(x)$ denote the vocabulary packet carrying a new token $x\in\Delta\mathcal{V}_f$, and define the prerequisite set
\begin{equation}
\mathcal{Q}_i=\left\{q(x):x\in\{s_i,r_i,o_i\}\cap\Delta\mathcal{V}_f\right\}.
\label{eq:prerequisite_set}
\end{equation}
The case $\mathcal{Q}_i=\emptyset$ means that all tokens of $t_i$ are already synchronized.

Let $\mathcal{U}_0$ denote the vocabulary-increment packet set, let $\mathcal{U}_t=\{1,\ldots,L\}$ denote the indexed-triplet packet set, and let $\mathcal{U}=\mathcal{U}_0\cup\mathcal{U}_t$. If $\Delta\mathcal{V}_f^{(q)}$ is the token subset in vocabulary packet $q$ and $A_{q,f}=1$ indicates successful recovery within the allowed HARQ depth, the acknowledged vocabulary state evolves as
\begin{equation}
\mathcal{V}_f=\mathcal{V}_{f-1}\cup\bigcup_{q\in\mathcal{U}_0:A_{q,f}=1}\Delta\mathcal{V}_f^{(q)}.
\label{eq:vocab_state_update}
\end{equation}
An unrecovered token remains unsynchronized and is announced again when referenced by a later frame, while the present optimizer is frame-wise conditional on $\mathcal{V}_{f-1}$ and does not assign a separate long-horizon value to future token reuse.

For $\mathcal{T}_f^{(\ell)}=\{t_i:\ell(i)=\ell\}$, the vocabulary and indexed-triplet payloads are
\begin{subequations}
\label{eq:packet_payloads}
\begin{align}
B_{0q}&=B_{\rm hdr}^{v}+\sum_{x\in\Delta\mathcal{V}_f^{(q)}}\left(b_{\rm tok}(x)+b_{\rm crc}^{\rm tok}\right),\label{eq:vocab_payload}\\
B_{\ell}&=B_{\rm hdr}^{t}+\sum_{i:\ell(i)=\ell}\left(3b_{\rm id}(f)+b_{\rm aux}\right).\label{eq:triplet_payload}
\end{align}
\end{subequations}
Here $B_{\rm hdr}^{v}$ and $B_{\rm hdr}^{t}$ contain packet metadata and packet-level cyclic redundancy check (CRC) fields, $b_{\rm tok}(x)$ is the compressed token-string length, $b_{\rm crc}^{\rm tok}$ is an optional token-verification field, and $b_{\rm aux}$ contains confidence or attribute information. Each triplet has a positive task weight $w_i$ fixed before wireless optimization and independent of realized packet-decoding outcomes, consistent with task-oriented communication formulations in which task importance is determined before channel realization \cite{ref20,ref23}.

\subsection{Wireless Link and HARQ Reliability}
\label{subsec:wireless_harq}

Each packet $u\in\mathcal{U}$ is assigned a finite wireless action
\begin{equation}
a_u=(m_u,P_u,K_u)\in\mathcal{A}_u,
\label{eq:packet_configuration}
\end{equation}
where $m_u\in\mathcal{M}$ is a rate/MCS label, $P_u\in\mathcal{P}$ is a transmit-power level, and $K_u\in\{1,\ldots,K_{\max}\}$ is the maximum HARQ depth. Rate adaptation and HARQ-depth selection are standard reliability-control dimensions in retransmission-aware wireless links \cite{ref9,ref11,ref12}. In the analytical branch, $m_u$ specifies the target spectral efficiency used by the outage model rather than a standard-specific code family.

For nominal spectral efficiency $\rho(m_u)$, the channel uses per round and actual coding rate are
\begin{equation}
n_u=\left\lceil\frac{B_u}{\rho(m_u)}\right\rceil,\qquad R_u=\frac{B_u}{n_u}\leq\rho(m_u).
\label{eq:blocklength_rate}
\end{equation}

The received signal in HARQ round $k$ is
\begin{equation}
y_{u,k}=\sqrt{P_u g(d)}h_{u,k}x_{u,k}+z_{u,k},
\label{eq:received_signal}
\end{equation}
where $x_{u,k}$ has unit average power, $h_{u,k}\sim\mathcal{CN}(0,1)$ is a circularly symmetric complex Gaussian fading coefficient, $g(d)$ is the large-scale gain, and $z_{u,k}\sim\mathcal{CN}(0,N_0W)$ is complex additive white Gaussian noise (AWGN) over bandwidth $W$ \cite{ref2,ref3}.

The distance-dependent path loss and receiver noise power in dB units are
\begin{align}
L(d)&=32.4+20\log_{10}(f_c)+10\alpha\log_{10}(d),\label{eq:path_loss}\\
N^{\rm dBm}&=-174+10\log_{10}(W)+F,\label{eq:noise_power}
\end{align}
where $f_c$ is in GHz, $d$ is in meters, $\alpha$ is the path-loss exponent, and $F$ is the receiver noise figure. The log-distance path-loss abstraction and thermal-noise model follow standard wireless link-budget practice \cite{ref2,ref3}. With log-normal shadowing $\xi\sim\mathcal{N}(0,\sigma_{\rm sh}^2)$, implementation loss $L_{\rm imp}$, and a conservative channel state information (CSI) backoff margin $M_{\rm CSI}\geq0$, the average signal-to-noise ratio (SNR) is
\begin{equation}
\bar{\gamma}_u=10^{\frac{P_u^{\rm dBm}-L(d)-\xi-L_{\rm imp}-M_{\rm CSI}-N^{\rm dBm}}{10}}.
\label{eq:average_snr}
\end{equation}
The CSI margin is zero in the default setting, and the frame-level large-scale state is denoted by $Z_f=(d,\xi,L_{\rm imp},M_{\rm CSI})$.

For Chase-combining HARQ, $\Gamma_u^{(k)}=\sum_{j=1}^{k}\gamma_{u,j}$ is the accumulated SNR after $k$ rounds, where the per-round instantaneous SNR $\gamma_{u,j}$ is exponentially distributed with mean $\bar{\gamma}_u$ under independent Rayleigh fading. Chase combining accumulates repeated observations of the same packet before decoding \cite{ref7,ref8}. Hence, $\Gamma_u^{(k)}$ follows a gamma distribution with shape parameter $k$ and scale parameter $\bar{\gamma}_u$. With decoding threshold $\theta_u=2^{R_u}-1$, the failure probability after $k$ rounds is the outage probability $\Pr\{\Gamma_u^{(k)}<\theta_u\}$, given by the cumulative distribution function of $\Gamma_u^{(k)}$ evaluated at $\theta_u$ \cite{ref4,ref7}:
\begin{equation}
\epsilon_u^{(k)}
=
1-\exp\!\left(-\frac{\theta_u}{\bar{\gamma}_u}\right)
\sum_{r=0}^{k-1}
\frac{1}{r!}
\left(\frac{\theta_u}{\bar{\gamma}_u}\right)^r,
\label{eq:harq_failure_k}
\end{equation}
and the residual failure probability associated with action $a_u$ is
\begin{equation}
\epsilon_u(a_u)=\epsilon_u^{(K_u)}.
\label{eq:residual_failure}
\end{equation}
The allocator requires only packet-level residual failure probability, expected HARQ rounds, delay, and energy. Hence, the analytical outage model acts as an interchangeable link interface: it can be replaced by BLER values obtained from a link-level or link-to-system abstraction without changing the allocation equations \cite{ref51,ref52}.

\subsection{Delay and Energy Accounting}
\label{subsec:delay_energy}

If $T_u$ is the actual number of HARQ rounds used by packet $u$, round $k$ is transmitted only when the first $k-1$ attempts fail, which gives the standard truncated-retransmission expectation \cite{ref8,ref10}
\begin{equation}
\bar{T}_u=\mathbb{E}[T_u]=\sum_{k=1}^{K_u}\epsilon_u^{(k-1)},\qquad \epsilon_u^{(0)}=1.
\label{eq:expected_rounds}
\end{equation}
The feedback-aware expected packet delay is
\begin{equation}
\bar{\tau}_u(a_u)=\bar{T}_u\left(\frac{n_u}{W}+\tau_{\rm fb}\right),
\label{eq:expected_delay}
\end{equation}
where $\tau_{\rm fb}$ is the ACK/NACK processing and feedback time per transmitted attempt. Feedback-aware retransmission delay and energy accounting are standard in latency-constrained HARQ design \cite{ref10,ref11}. The expected energy is
\begin{align}
\bar{E}_u(a_u)={}&P_u^{\rm W}\bar{T}_u\frac{n_u}{W}+P_c\bar{T}_u\left(\frac{n_u}{W}+\tau_{\rm fb}\right)\nonumber\\
&+P_{\rm fb}\tau_{\rm fb}\bar{T}_u,
\label{eq:expected_energy}
\end{align}
where $P_u^{\rm W}$ is transmit power in watts, $P_c$ is the circuit-equivalent power, and $P_{\rm fb}$ is the feedback-processing equivalent power.

\section{Dependency-Aware Reliability Allocation}
\label{sec:allocation}

\subsection{Semantic Metric and Allocation Problem}
\label{subsec:semantic_problem}

Triplet $t_i$ is interpretable only when its indexed-triplet packet and every prerequisite vocabulary packet are recovered. This receiver-side utility definition follows the task-oriented principle that successful communication is determined by usable semantic content rather than packet recovery in isolation \cite{ref20,ref23,ref25}. Under conditional packet independence given $Z_f$, its interpretability probability is
\begin{equation}
p_i^{\rm int}(Z_f)=\left(1-\epsilon_{\ell(i)}(Z_f)\right)\prod_{q\in\mathcal{Q}_i}\left(1-\epsilon_q(Z_f)\right),
\label{eq:interpretability_probability}
\end{equation}
where an empty product equals one. The expected dependency-aware semantic distortion and its normalized form are
\begin{align}
\bar{D}_f&=\sum_{i=1}^{N_f}w_i\left[1-\left(1-\epsilon_{\ell(i)}\right)\prod_{q\in\mathcal{Q}_i}\left(1-\epsilon_q\right)\right],\label{eq:semantic_distortion}\\
\bar{D}_f^{\rm norm}&=\frac{\bar{D}_f}{\sum_{i=1}^{N_f}w_i}.\label{eq:normalized_distortion}
\end{align}

For a multi-fact query $h$ requiring triplets $\mathcal{C}_{f,h}$, define the distinct required-packet set $\mathcal{R}_{f,h}=\{\ell(i):i\in\mathcal{C}_{f,h}\}\cup\bigcup_{i\in\mathcal{C}_{f,h}}\mathcal{Q}_i$. Its communication-layer delivery probability is
\begin{equation}
p_h^{\rm qry}=\prod_{u\in\mathcal{R}_{f,h}}(1-\epsilon_u).
\label{eq:query_delivery}
\end{equation}
The optimization below uses $\bar D_f$, while the experiments also report exact multi-fact query delivery to expose cases in which the triplet-weighted objective and joint query delivery prefer different allocations.

The distortion separates into indexed-triplet and vocabulary-induced components:
\begin{subequations}
\label{eq:failure_decomposition}
\begin{align}
\bar{D}_f^{\rm trip}&=\sum_{i=1}^{N_f}w_i\epsilon_{\ell(i)},\label{eq:triplet_failure_component}\\
\bar{D}_f^{\rm vocab}&=\sum_{i=1}^{N_f}w_i(1-\epsilon_{\ell(i)})\left[1-\prod_{q\in\mathcal{Q}_i}(1-\epsilon_q)\right].\label{eq:vocab_failure_component}
\end{align}
\end{subequations}
Thus, $\bar D_f=\bar D_f^{\rm trip}+\bar D_f^{\rm vocab}$, and the vocabulary-induced term vanishes when every $\mathcal{Q}_i$ is empty.

The product model in \eqref{eq:interpretability_probability} can be separated from the wireless optimizer. If the true joint recovery probability is $p_i^{\rm joint}=p_i^{\rm int}+\Delta_i$, then
\begin{equation}
\left|\bar{D}_f^{\rm joint}-\bar{D}_f\right|\leq\sum_{i=1}^{N_f}w_i|\Delta_i|,
\label{eq:mismatch_certificate}
\end{equation}
which makes the effect of packet-dependence mismatch explicit without changing the packet-level link model.

The wireless controller selects one action for every packet by solving
\begin{subequations}
\label{prob:p0}
\begin{align}
\mathbf{P0}:\quad \min_{\{a_u\}_{u\in\mathcal{U}}}\quad &\bar{D}_f\!\left(\{a_u\}_{u\in\mathcal{U}}\right)\label{prob:p0_objective}\\
\text{s.t.}\quad &\sum_{u\in\mathcal{U}}\bar{\tau}_u(a_u)\leq\tau_{\max},\label{prob:p0_delay}\\
&\sum_{u\in\mathcal{U}}\bar{E}_u(a_u)\leq E_{\max},\label{prob:p0_energy}\\
&\epsilon_q(a_q)\leq\delta_v,\qquad q\in\mathcal{U}_0,\label{prob:p0_vocab}\\
&a_u\in\mathcal{A}_u,\qquad u\in\mathcal{U}.\label{prob:p0_candidate}
\end{align}
\end{subequations}
The first two constraints impose expected frame-delay and frame-energy budgets, while $\delta_v$ is an optional common reliability floor for vocabulary packets and is inactive in the default experiments with $\delta_v=1$. Problem $\mathbf{P0}$ is finite and generally nonconvex because its variables are discrete and its objective contains prerequisite-induced products of packet success probabilities.

\subsection{Dependency Structure and Link-Adaptation Rule}
\label{subsec:dependency_link}

The distortion is affine in the residual failure probability of any one packet when all other packet reliabilities are fixed. For indexed-triplet packet $\ell\in\mathcal{U}_t$ and vocabulary packet $q\in\mathcal{U}_0$, define
\begin{subequations}
\label{eq:dependency_coefficients}
\begin{align}
\alpha_{\ell}&=\sum_{i:\ell(i)=\ell}w_i\prod_{q\in\mathcal{Q}_i}(1-\epsilon_q),\label{eq:alpha_triplet}\\
\alpha_q&=\sum_{i:q\in\mathcal{Q}_i}w_i(1-\epsilon_{\ell(i)})\prod_{q'\in\mathcal{Q}_i\setminus\{q\}}(1-\epsilon_{q'}).\label{eq:alpha_vocab}
\end{align}
\end{subequations}
The coefficient $\alpha_{\ell}$ is the current weighted value of delivering indexed packet $\ell$, whereas $\alpha_q$ is the current enabling value of vocabulary packet $q$ and varies with the reliabilities of its downstream triplets and other prerequisites. For any packet $u$, fixing the other packet configurations gives
\begin{equation}
\bar{D}_f=C_u+\alpha_u\epsilon_u(a_u),
\label{eq:block_affine_form}
\end{equation}
where $C_u$ is independent of $a_u$.

Static UEP models allocate protection using reliability-independent source importance weights \cite{ref13,ref14}. The following proposition identifies why an active prerequisite graph cannot generally be reduced to that form.

\begin{proposition}[Non-equivalence to static separable UEP]
\label{prop:nonseparable}
Assume that at least one positive-weight triplet $i$ has a prerequisite $q\in\mathcal{Q}_i$. On any open reliability domain where the success probabilities of the remaining prerequisites are positive, there are no reliability-independent constants $C$ and $\{\omega_u\}$ such that
\begin{equation}
\bar{D}_f(\boldsymbol{\epsilon})=C+\sum_{u\in\mathcal{U}}\omega_u\epsilon_u
\label{eq:static_separable_form}
\end{equation}
throughout that domain.
\end{proposition}

\begin{IEEEproof}
For any indexed packet $\ell$ and vocabulary packet $q$ sharing at least one dependent triplet,
\begin{equation}
\frac{\partial^2\bar{D}_f}{\partial\epsilon_{\ell}\partial\epsilon_q}=-\!\sum_{i:\ell(i)=\ell,\,q\in\mathcal{Q}_i}w_i\!\prod_{q'\in\mathcal{Q}_i\setminus\{q\}}(1-\epsilon_{q'})<0.
\label{eq:mixed_partial}
\end{equation}
Every mixed derivative of \eqref{eq:static_separable_form} is zero, which gives a contradiction.
\end{IEEEproof}

For multipliers $\lambda,\mu\geq0$, define
\begin{align}
\Phi(\mathbf{a};\lambda,\mu)={}&\bar{D}_f(\mathbf{a})+\lambda\left[\sum_{u\in\mathcal{U}}\bar{\tau}_u(a_u)-\tau_{\max}\right]\nonumber\\
&+\mu\left[\sum_{u\in\mathcal{U}}\bar{E}_u(a_u)-E_{\max}\right].
\label{eq:lagrangian_objective}
\end{align}
After removing vocabulary actions that violate \eqref{prob:p0_vocab}, fixing all other packet configurations and the multipliers yields the exact packet update
\begin{equation}
a_u^{\star}=\arg\min_{a\in\mathcal{A}_u^{\rm adm}}\left\{\alpha_u\epsilon_u(a)+\lambda\bar{\tau}_u(a)+\mu\bar{E}_u(a)\right\}.
\label{eq:block_update}
\end{equation}

The wireless meaning of $\alpha_u$ follows directly from \eqref{eq:block_update}. For two actions $a$ and $b$ with $\epsilon_u(a)<\epsilon_u(b)$, define $c_u(a)=\lambda\bar{\tau}_u(a)+\mu\bar{E}_u(a)$. The more reliable action $a$ is preferred whenever
\begin{equation}
\alpha_u\geq\eta_u(a,b)\triangleq\frac{c_u(a)-c_u(b)}{\epsilon_u(b)-\epsilon_u(a)}.
\label{eq:switching_threshold}
\end{equation}
If $c_u(a)\leq c_u(b)$, action $a$ weakly dominates $b$ for every $\alpha_u\geq0$; otherwise, \eqref{eq:switching_threshold} is the reliability-value threshold at which its additional priced delay and energy become worthwhile. Because $\alpha_q$ in \eqref{eq:alpha_vocab} changes with prerequisite state, the selected rate/MCS, transmit power, or HARQ depth can change even when the payload, channel state, candidate menu, and resource prices are fixed.

\subsection{Algorithm and Properties}
\label{subsec:algorithm_properties}

\begin{algorithm}[!t]
\caption{Dependency-Aware Lagrangian Link Adaptation}
\label{alg:block}
\begin{algorithmic}[1]
\REQUIRE Triplets $\mathcal{T}_f$; dependency sets $\{\mathcal{Q}_i\}$; candidate sets $\{\mathcal{A}_u\}$; budgets $\tau_{\max}$ and $E_{\max}$; optional floor $\delta_v$; maximum sweeps $R_{\max}$.
\STATE Screen the candidate sets to form $\mathcal{A}_u^{\rm adm}$ and return ``infeasible'' if any admissible set is empty.
\STATE Construct deterministic admissible seeds and initialize the feasible and least-violation records.
\STATE Initialize one multiplier start $(\lambda,\mu)$.
\FOR{$r=1,2,\ldots,R_{\max}$}
\FOR{each packet $u\in\mathcal{U}$}
\STATE Compute $\alpha_u$ from \eqref{eq:alpha_triplet} or \eqref{eq:alpha_vocab}.
\STATE Update $a_u$ by \eqref{eq:block_update}.
\ENDFOR
\STATE Compute $\tau_{\rm tot}=\sum_u\bar{\tau}_u(a_u)$ and $E_{\rm tot}=\sum_u\bar{E}_u(a_u)$.
\STATE Update the least-violation diagnostic record.
\STATE $\lambda\leftarrow[\lambda+s_r(\tau_{\rm tot}-\tau_{\max})]^+$.
\STATE $\mu\leftarrow[\mu+s_r(E_{\rm tot}-E_{\max})]^+$.
\IF{$\tau_{\rm tot}\leq\tau_{\max}$ and $E_{\rm tot}\leq E_{\max}$}
\STATE Retain the current configuration if it improves the feasible incumbent.
\ENDIF
\ENDFOR
\IF{a feasible configuration was found}
\RETURN the best feasible configuration.
\ELSE
\RETURN the least-violation configuration with status ``no feasible configuration found''.
\ENDIF
\end{algorithmic}
\end{algorithm}

Here $[x]^+=\max\{x,0\}$ and the implementation uses the diminishing step size $s_r=s_0/\sqrt{r}$. A current-point-preserving tie rule is used in \eqref{eq:block_update}, so a tied candidate does not trigger a configuration change. The numerical implementation repeats the sweep for the multiplier starts listed in Table~\ref{tab:params} and retains the best feasible configuration encountered.

\begin{proposition}[Finite fixed-multiplier termination]
\label{prop:fixed_termination}
For fixed $\lambda$ and $\mu$, finite admissible candidate sets, cyclic packet updates, and the current-point-preserving tie rule, repeated packet updates reach after finitely many strict configuration changes a point that minimizes $\Phi(\mathbf{a};\lambda,\mu)$ with respect to every individual packet block.
\end{proposition}

\begin{IEEEproof}
By \eqref{eq:block_affine_form}, update \eqref{eq:block_update} minimizes the complete packet-dependent part of $\Phi$. Every accepted configuration change therefore strictly decreases $\Phi$, and the Cartesian product of the finite candidate sets contains only finitely many configurations. When no further strict change is available, no packet admits a strictly improving unilateral update.
\end{IEEEproof}

Proposition~\ref{prop:fixed_termination} characterizes the packet-update operator for fixed $\lambda$ and $\mu$. Algorithm~\ref{alg:block} interleaves one cyclic packet sweep with a dual-price update and is therefore used as a finite-action primal--dual search for $\mathbf{P0}$. Its solution quality is assessed against exact enumeration on tractable instances in Section~\ref{subsec:results_setup}. Likewise, failure of the heuristic search to encounter a feasible configuration is not a certificate that $\mathbf{P0}$ itself is infeasible; only an empty admissible packet candidate set provides the immediate infeasibility certificate in Step~1. Let $U=|\mathcal{U}|$ and $A=\max_u|\mathcal{A}_u|$. One sweep costs $O(UAK_{\max}+UN_f)$ when link metrics are evaluated online and $O(UA+UN_f)$ when reliability, delay, and energy are precomputed for each candidate.

\section{Numerical and Trace-Driven Results}
\label{sec:results}

\subsection{Experimental Setup, Baselines, and Solver Validation}
\label{subsec:results_setup}

The synthetic experiments vary distance, expected frame-delay and frame-energy budgets, shadowing strength, and conservative CSI margin while using the same semantic packetization and the same rate/MCS--power--HARQ action menu across policies. The analytical link follows the Rayleigh and Chase-combining model in Section~\ref{subsec:wireless_harq}. The GQA experiments map validation scene graphs and grounded semantic programs to triplet packets, vocabulary increments, and question-specific prerequisite sets using the public GQA annotations \cite{ref48}. For each evaluated image, five non-paraphrase grounded questions are selected, and $q_i$ denotes the number of those questions whose grounded programs require triplet $i$; the primary task weight is $w_i=1+q_i/5$.

The controlled link-adaptation experiment used in Fig.~\ref{fig:robust}(c) fixes a 112-bit vocabulary packet, the analytical 500-m Rayleigh Chase-combining HARQ (HARQ-CC) link, the full $4\times4\times4=64$ action menu, and the dual prices $\lambda=8.3488\times10^{-3}$ and $\mu=1.4427\times10^{-2}$ while varying only the dependency coefficient $\alpha_q$. Thus, any change in the selected action in Fig.~\ref{fig:robust}(c) is attributable to a change in the reliability value assigned to the prerequisite packet under fixed wireless conditions and fixed resource prices.

\begin{table}[!t]
\caption{Wireless, Candidate-Menu, and Solver Parameters}
\label{tab:params}
\centering
\scriptsize
\setlength{\tabcolsep}{2.2pt}
\begin{tabular}{p{0.37\linewidth}p{0.57\linewidth}}
\toprule
Parameter & Value\\
\midrule
Carrier / bandwidth & $3.5$ GHz / $180$ kHz\\
Noise power spectral density / noise figure & $-174$ dBm/Hz / $7$ dB\\
Path-loss exponent / shadowing std. & $3.15$ / $5$ dB\\
Implementation loss $L_{\rm imp}$ & $2$ dB\\
Power menu & $\{-5,5,15,23\}$ dBm\\
MCS efficiencies & $\{0.50,0.80,1.20,1.80\}$ bpcu\\
HARQ depth & $K_{\max}=4$, $K_{\rm UEP}=2$\\
Feedback / circuit / feedback power & $0.4$ ms / $28$ mW / $4$ mW\\
Expected delay / energy budgets & $38$ ms / $5.2$ mJ per frame\\
Vocabulary floor & $\delta_v=1$ (inactive)\\
Packet headers & vocab./triplet: $48/64$ bits\\
Token/auxiliary fields & $b_{\rm tok}(x)=12+8\min\{|x|,12\}$; $b_{\rm crc}^{\rm tok}=0$, $b_{\rm aux}=10$ bits\\
Packet grouping & at most 6 new tokens / 6 triplets per packet\\
Multiplier schedule & $s_r=0.75/\sqrt{r}$\\
Iteration cap / initializations & $500$ per start / three starts\\
Feasible polishing & Six one-block passes\\
\bottomrule
\end{tabular}
\end{table}

\begin{table}[!t]
\caption{GQA Trace-Driven Evaluation Protocol}
\label{tab:gqa_protocol}
\centering
\scriptsize
\setlength{\tabcolsep}{3pt}
\renewcommand{\arraystretch}{1.08}
\begin{tabular}{p{0.34\linewidth}p{0.58\linewidth}}
\toprule
\textbf{Item} & \textbf{Setting}\\
\midrule
Source annotations &
10,696 scene graphs; 132,062 balanced questions\\
Eligible grounded questions &
104,128 (78.85\%)\\
Session protocol &
$20\times(5$ warm-up $+50$ evaluated images$)$\\
Evaluated workload &
1,000 images; 5,000 questions\\
Question composition &
2,002 rel.; 1,693 attr.; 731 obj.; 574 cat.\\
Maximum triplets per frame &
48\\
Primary task weight &
$w_i=1+q_i/5$\\
Primary / cross-check link &
BLER ledger / analytical Rayleigh HARQ-CC\\
Statistical replicate &
20 independent sessions\\
Uncertainty summary &
Two-sided Student-$t$ 95\% confidence interval of the session mean\\
\bottomrule
\end{tabular}
\end{table}

All policies use identical semantic frames, packetization, payloads, wireless link ledgers, rate/MCS labels, power levels, delay and energy models, and frame budgets. For the fixed packet-weight baselines, define
\begin{equation}
J_{\rm pkt}(\mathbf a)=\sum_{\ell\in\mathcal U_t}\omega_\ell\epsilon_\ell(a_\ell)+\sum_{q\in\mathcal U_0}\omega_q\epsilon_q(a_q),\qquad
\omega_\ell=\sum_{i:\ell(i)=\ell}w_i,
\label{eq:baseline_packet_objective}
\end{equation}
where every vocabulary weight $\omega_q$ is fixed before link optimization and is independent of the realized packet reliabilities. To isolate graph awareness from reliability-state adaptation, Static-DA UEP uses
\begin{equation}
\omega_{\ell}^{\rm SDA}=\sum_{i:\ell(i)=\ell}w_i,\qquad
\omega_q^{\rm SDA}=\sum_{i:q\in\mathcal Q_i}w_i.
\label{eq:static_da_weights}
\end{equation}
These coefficients equal the proposed marginal values at the all-success reference state, $\alpha_u(\boldsymbol\epsilon_{-u}=\mathbf 0)$, and are then frozen before wireless optimization. Hence Static-DA uses the same prerequisite graph, task weights, and full wireless action menu as the proposed method but removes reliability-state adaptation. Dependency-agnostic HARQ uses the same indexed-packet aggregate task weights but assigns fixed vocabulary weights without using downstream prerequisite sets. Semantic-priority HARQ uses a fixed semantic ordering and a fixed protocol-overhead priority for vocabulary packets. Priority UEP uses the same fixed priority rule but restricts $K_u\le K_{\rm UEP}$; UEP-only fixes $K_u=1$ while retaining unequal rate and power choices; Uniform HARQ exhaustively selects one common $(m,P,K)$ tuple for all packets. These definitions are summarized again in Appendix~\ref{app:baseline_details}.

For the GQA statistics, frames within a session are averaged first because the synchronized vocabulary evolves across frames, and the 20 independent session means are then used to form two-sided Student-$t$ 95\% confidence intervals. To test sensitivity to the analytical outage model, we also use a separate non-analytic packet-level BLER ledger, consistent with lookup-based physical-layer (PHY) abstractions used in system-level studies \cite{ref51,ref52}. It contains 1,152 anchor entries over $B\in\{112,360,760\}$ bits, $\rho\in\{0.50,0.80,1.20,1.80\}$ bpcu, $P\in\{-5,5,15,23\}$ dBm, $d\in\{80,160,240,320,410,500\}$ m, and $K\in\{1,2,3,4\}$. The entries are treated as an external reliability map rather than as a calibrated implementation of a particular PHY.

\begin{figure*}[!t]
\centering
\includegraphics[width=0.95\textwidth]{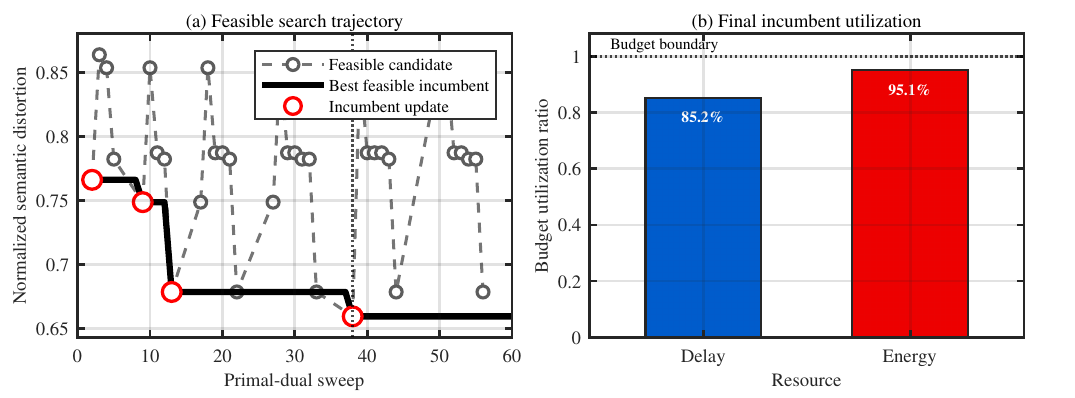}
\caption{Solver diagnostics for the retained analytical run. (a) Feasible candidate distortion across primal--dual sweeps and the best feasible incumbent; markers denote incumbent updates. (b) Final incumbent delay and energy utilization relative to the corresponding budgets.}
\label{fig:convergence}
\end{figure*}

Fig.~\ref{fig:convergence}(a) shows that the feasible incumbent improves in discrete steps as the finite action assignment changes; the last displayed incumbent update occurs at sweep 38. Fig.~\ref{fig:convergence}(b) shows that the retained incumbent uses 85.2\% of the delay budget and 95.1\% of the energy budget. These diagnostics support feasibility of the reported incumbent but do not constitute a global-convergence claim for the complete primal--dual procedure.

\begin{table}[!t]
\caption{Small-Scale Exhaustive Global-Optimum Audit}
\label{tab:global_audit}
\centering
\scriptsize
\setlength{\tabcolsep}{2.4pt}
\begin{tabular}{rcccc}
\toprule
$U$ & Instances & Comb./inst. & Global hits & Mean / worst $g_{\rm rel}$\\
\midrule
4 & 10 & 4,096 & 10/10 & 0 / 0\\
5 & 10 & 32,768 & 10/10 & 0 / 0\\
6 & 10 & 262,144 & 8/10 & 0.157\% / 1.095\%\\
\midrule
All & 30 & -- & 28/30 & 0.052\% / 1.095\%\\
\bottomrule
\end{tabular}
\end{table}

To assess solution quality against the exact optimum, Table~\ref{tab:global_audit} uses the analytical 500-m link, the same 38-ms and 5.2-mJ budgets, and a reduced eight-action menu with $\rho\in\{0.5,1.2\}$ bpcu, $P\in\{15,23\}$ dBm, and $K\in\{1,2\}$. Complete enumeration and the proposed allocation agree in 28 of 30 instances. Across all 30 cases, the mean relative gap is $0.052\%$ and the worst observed gap is $1.095\%$. These exact comparisons quantify the empirical solution quality of the finite-action allocation on instances for which exhaustive optimization is tractable.

\subsection{Main Wireless Performance and Resource Tradeoffs}
\label{subsec:wireless_tradeoffs}

\begin{figure*}[!t]
\centering
\includegraphics[width=0.95\textwidth]{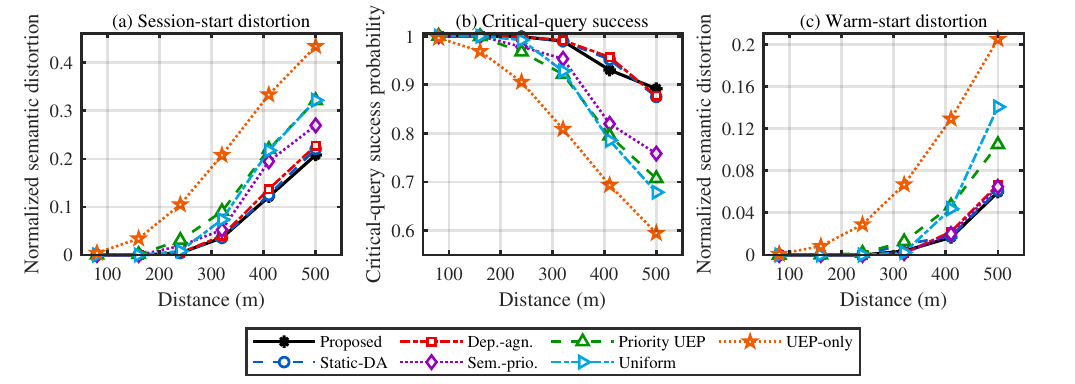}
\caption{Distance-domain analytical-link evidence. (a) Session-start normalized semantic distortion. (b) Critical-query success probability. (c) Warm-start normalized semantic distortion. The panels use the same policy set shown in (a).}
\label{fig:main}
\end{figure*}

Fig.~\ref{fig:main} isolates how prerequisite activity interacts with distance-dependent link degradation. At short distances, the full-menu policies cluster near low distortion and high query success because most packets are already highly reliable. As distance increases, session-start distortion rises and critical-query success falls, while the proposed policy remains on the low-distortion/high-success envelope of the displayed policies. Static-DA stays close to the proposed curve over part of the sweep, whereas fixed-priority and restricted-HARQ baselines separate more strongly as the link degrades. In the warm-start case of Fig.~\ref{fig:main}(c), the separation among the full-menu policies is substantially reduced because fewer triplets depend on newly transmitted vocabulary mappings. The contrast between session start and warm start therefore supports the interpretation that the benefit is tied to active prerequisites rather than to an expanded wireless action menu.

\begin{figure*}[!t]
\centering
\includegraphics[width=0.95\textwidth]{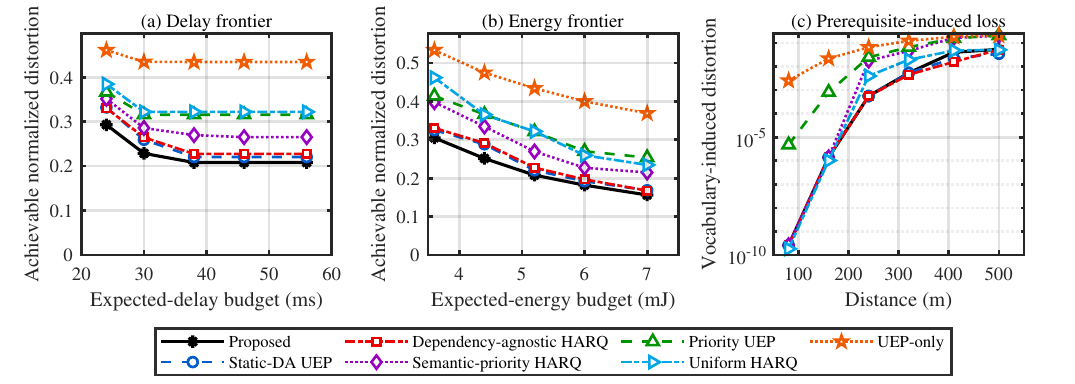}
\caption{Resource and prerequisite behavior under the analytical link. (a) Achievable normalized distortion versus expected frame-delay budget. (b) Achievable normalized distortion versus expected frame-energy budget. (c) Vocabulary-induced distortion versus distance on a logarithmic scale.}
\label{fig:resource}
\end{figure*}

Figs.~\ref{fig:resource}(a) and \ref{fig:resource}(b) show the reliability--resource tradeoff generated by finite rate/MCS, power, and HARQ-depth choices. The piecewise-flat behavior follows from the discrete action menu: additional delay or energy changes the distortion only when a different finite action assignment becomes preferable or feasible. Across the displayed sweeps, the proposed method tracks the lowest-distortion envelope, while the gaps to Static-DA and the fixed-weight baselines depend on the available resource budget. Fig.~\ref{fig:resource}(c) isolates the prerequisite-induced component and shows that vocabulary-induced loss becomes increasingly important as distance grows. This is the component directly affected by the state-dependent vocabulary coefficient in \eqref{eq:alpha_vocab}.

\subsection{Dependency Mechanism and GQA Trace Evidence}
\label{subsec:dependency_gqa}

Static-DA is deliberately stronger than payload-only or protocol-priority baselines because it is given the exact prerequisite graph and the same task weights as the proposed method. The only removed ingredient is the dependence of the packet coefficient on the current reliabilities of neighboring packets. Table~\ref{tab:static_da_primary} therefore provides a targeted test of the structural difference established in Proposition~\ref{prop:nonseparable}. The audit uses a separately retained deterministic stress ensemble, so its values are interpreted only within that stated ensemble and are not used as point values for the distance curves in Figs.~\ref{fig:main} or \ref{fig:calibrated}.

\begin{table}[!t]
\caption{Strong Static-DA Audit at 500 m (Independent Stress Ensemble)}
\label{tab:static_da_primary}
\centering
\scriptsize
\setlength{\tabcolsep}{3.0pt}
\begin{tabular}{lccc}
\toprule
Reliability map & Proposed & Static-DA & Dep.-agn.\\
\midrule
Analytical Rayleigh & $1.565{\times}10^{-4}$ & $1.565{\times}10^{-4}$ & $1.581{\times}10^{-4}$\\
Table-driven BLER & $0.1983$ & $0.2214$ & $0.2448$\\
\bottomrule
\end{tabular}
\vspace{1pt}
\begin{minipage}{0.96\linewidth}\scriptsize
Entries are mean normalized semantic distortion. The analytical and BLER rows use 20 independently seeded frames; all three policies use the same 64-action menu, 38-ms expected-delay budget, and 5.2-mJ energy budget.
\end{minipage}
\end{table}

Under the analytical Rayleigh map in Table~\ref{tab:static_da_primary}, Proposed and Static-DA have the same reported mean distortion, while dependency-agnostic HARQ is only slightly higher. Under the table-driven BLER map, Proposed reduces the reported mean distortion by 10.43\% relative to Static-DA and by 19.00\% relative to dependency-agnostic HARQ. The contrast is consistent with the structural claim: graph awareness alone can be sufficient in a highly reliable operating regime, whereas a heterogeneous reliability map can make the state dependence of the marginal packet value consequential. The table does not support a claim of pointwise dominance over Static-DA at every operating point.

\begin{figure}[!t]
\centering
\includegraphics[width=\linewidth]{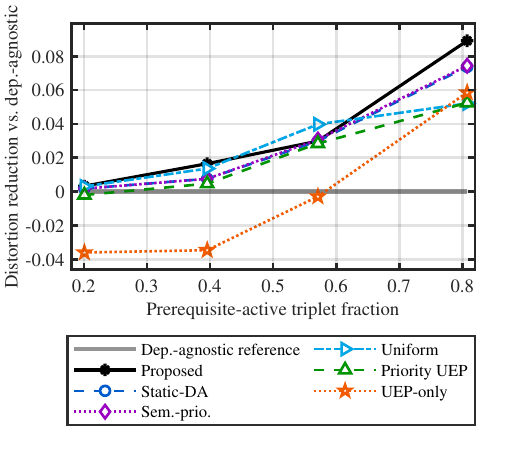}
\caption{GQA dependency-activity cross-check. Distortion reduction relative to dependency-agnostic HARQ is plotted versus the prerequisite-active triplet fraction; positive values indicate lower distortion than the dependency-agnostic reference.}
\label{fig:gqa_dependency}
\end{figure}

Fig.~\ref{fig:gqa_dependency} provides a trace-driven mechanism check that does not rely on a single 500-m aggregate. The proposed curve exhibits a larger positive distortion reduction as the prerequisite-active triplet fraction increases, whereas the gains are small in the lowest-activity region. This trend is the empirical signature expected from \eqref{eq:alpha_vocab}: when few facts depend on unsynchronized vocabulary, the additional state dependence has little room to change protection; when prerequisite activity is high, the reliability value of vocabulary packets can change materially with neighboring packet reliabilities.

\begin{table*}[!t]
\caption{GQA Results at 500 m Using the BLER Link Abstraction}
\label{tab:gqa_main}
\centering
\scriptsize
\setlength{\tabcolsep}{2.6pt}
\begin{tabular}{lcccccc}
\toprule
Policy & Distortion & Query delivery & Vocab. fail & Delay (ms) & Energy (mJ) & Feas. fraction\\
\midrule
Proposed & 0.00943 $\pm$ 0.00176 & 0.99147 $\pm$ 0.00346 & 0.00563 & 26.44 & 5.080 & 1.00\\
Dep.-agn. HARQ & 0.02284 $\pm$ 0.00546 & 0.98015 $\pm$ 0.00456 & 0.01420 & 26.54 & 5.066 & 1.00\\
Semantic-priority & 0.02927 $\pm$ 0.00645 & 0.96434 $\pm$ 0.00702 & 0.02881 & 26.28 & 5.080 & 1.00\\
Priority UEP & 0.03182 $\pm$ 0.00634 & 0.96335 $\pm$ 0.00716 & 0.02776 & 26.03 & 5.078 & 1.00\\
Uniform HARQ & 0.04134 $\pm$ 0.00870 & 0.95280 $\pm$ 0.00923 & 0.01228 & 23.71 & 4.614 & 1.00\\
UEP-only & 0.07384 $\pm$ 0.00493 & 0.93663 $\pm$ 0.00616 & 0.03544 & 25.54 & 5.061 & 1.00\\
\bottomrule
\end{tabular}
\end{table*}

The session-level GQA statistics in Table~\ref{tab:gqa_main} provide the primary quantitative trace-driven comparison. Relative to dependency-agnostic HARQ, the proposed allocation reduces the reported mean distortion from $0.02284$ to $0.00943$, corresponding to a 58.71\% reduction. Query failure decreases from $1-0.98015=0.01985$ to $1-0.99147=0.00853$, corresponding to a 57.03\% relative reduction. The proposed policy also reports a lower vocabulary-failure component, $0.00563$ versus $0.01420$, while both policies remain feasible on every evaluated frame. The GQA experiment therefore evaluates communication-layer delivery of the grounded semantic inputs; downstream visual question answering (VQA) inference is outside the evaluation pipeline.

\subsection{Wireless Uncertainty, Action Switching, and Link-Abstraction Cross-Checks}
\label{subsec:robustness_link}

\begin{figure*}[!t]
\centering
\includegraphics[width=\textwidth]{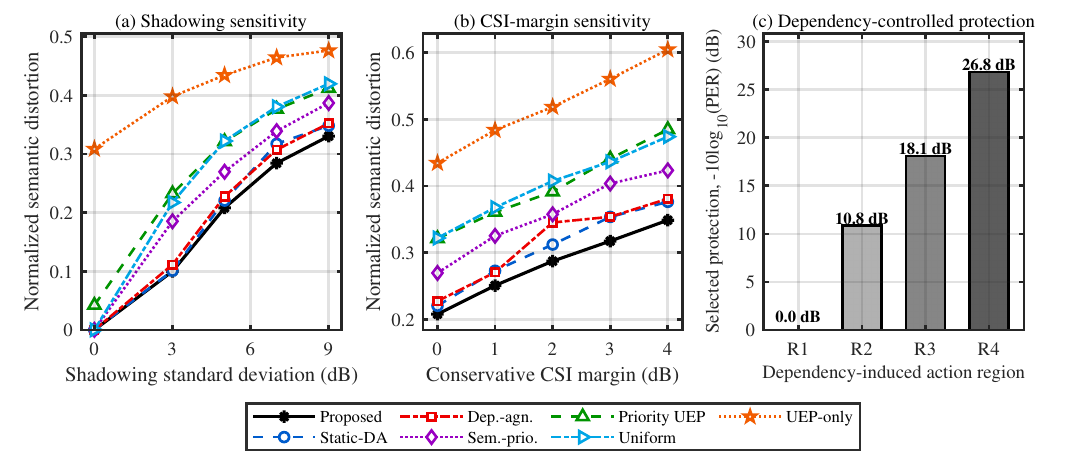}
\caption{Wireless uncertainty and dependency-controlled link adaptation. (a) Normalized semantic distortion versus shadowing standard deviation at 500 m. (b) Normalized semantic distortion versus conservative CSI margin. (c) (c) Selected packet protection based on the packet error rate (PER),
expressed as $-10\log_{10}(\mathrm{PER})$, where $\mathrm{PER}$ is the packet error rate, across the four observed dependency-induced action regions for a 112-bit vocabulary packet; the bar annotations give the selected $(\rho,P_{\rm dBm},K)$ tuple.}
\label{fig:robust}
\end{figure*}

Figs.~\ref{fig:robust}(a) and \ref{fig:robust}(b) show the expected degradation as the link is made more uncertain through stronger shadowing or a larger conservative CSI margin. The finite action menu creates plateaus and crossings among the baselines, so the evidence is interpreted as a sensitivity test rather than a claim of uniform pointwise dominance. Across the displayed range, the proposed curve remains on the low-distortion envelope.

Fig.~\ref{fig:robust}(c) directly visualizes the mechanism in \eqref{eq:switching_threshold}. With payload, channel state, action menu, and dual prices fixed, varying only $\alpha_q$ yields four observed action regions. The selected protection levels are 0.0, 10.8, 18.1, and 26.8 dB, with the annotated tuples changing from $(1.8,-5,1)$ to $(1.2,23,4)$, $(0.8,23,4)$, and $(0.5,23,4)$. Thus, a change in prerequisite-induced reliability value alone is sufficient in the controlled experiment to change the selected rate/MCS, transmit power, and/or HARQ depth. The figure demonstrates finite-action switching over the tested range; it does not assert that every possible $\alpha_q$ interval or every candidate action appears.

\begin{figure*}[!t]
\centering
\includegraphics[width=\textwidth]{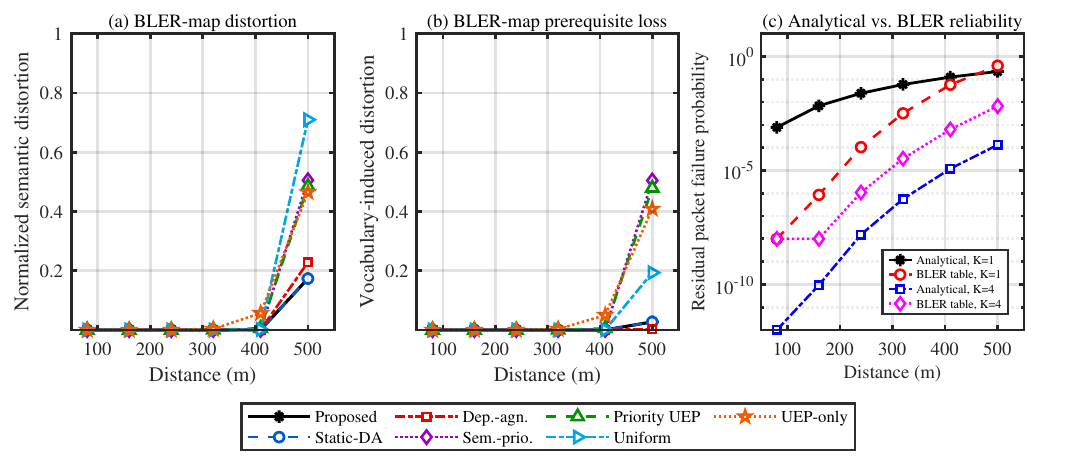}
\caption{Table-driven BLER link-abstraction cross-check. (a) Normalized semantic distortion versus distance. (b) Vocabulary-induced distortion versus distance. (c) Residual packet-failure probability versus distance for analytical Rayleigh HARQ-CC and the external BLER ledger at $K=1$ and $K=4$.}
\label{fig:calibrated}
\end{figure*}

Figs.~\ref{fig:calibrated}(a) and \ref{fig:calibrated}(b) repeat the distance-domain comparison after replacing the analytical residual-failure map by the non-analytic packet-level BLER ledger. The curves remain close to zero over the easier part of the distance range and separate sharply as the link approaches the difficult long-distance regime. The proposed method remains on the low-distortion envelope, and the vocabulary-induced component in Fig.~\ref{fig:calibrated}(b) shows that prerequisite loss contributes materially to the long-distance degradation. Fig.~\ref{fig:calibrated}(c) confirms that the analytical and table-driven reliability maps differ substantially for both $K=1$ and $K=4$ over the displayed distance range. The persistence of the dependency-aware trend across these materially different reliability maps supports robustness of the allocation mechanism to the choice of packet-level reliability model.

Across the analytical link, GQA traces, controlled switching experiment, and BLER-map cross-check, the results connect the observed performance differences to the same mechanism: active prerequisites alter the marginal reliability coefficient, and the resulting coefficient can change the selected MCS--power--HARQ action. The GQA results further show improved grounded-query delivery under the tested delay and energy budgets.

\section{Conclusion}
\label{sec:conclusion}

This paper studied reliability allocation for open-vocabulary scene-graph packets over a latency- and energy-constrained wireless visual uplink. Modeling vocabulary-increment packets as prerequisites of indexed triplets yields a multi-affine semantic distortion with mixed reliability terms, so packet value cannot generally be represented by reliability-independent static UEP weights. The resulting state-dependent dependency coefficient enters the packet-wise Lagrangian objective and determines when a different finite rate/MCS--power--HARQ action becomes preferable.

The proposed block method performs exact finite-candidate packet updates for fixed multipliers and has finite fixed-multiplier termination. In the reduced exhaustive audit, the implemented solver matches the global optimum in 28 of 30 cases, with a mean relative gap of $0.052\%$ and a worst observed gap of $1.095\%$. The independent Static-DA audit shows identical reported mean distortion for Proposed and Static-DA under the tested analytical Rayleigh setting, but lower distortion for Proposed under the table-driven BLER map, indicating that graph awareness and reliability-state adaptation are distinct ingredients whose practical separation depends on the reliability regime. The GQA session statistics further show a 58.71\% reduction in reported mean semantic distortion and a 57.03\% reduction in exact grounded-question failure relative to dependency-agnostic HARQ under the stated protocol. The dependency-activity, uncertainty, controlled switching, and link-abstraction figures connect these performance results back to the proposed mechanism: prerequisite activity changes packet reliability value, and that value can change the selected wireless action.

The current evaluation covers finite action menus, an analytical fading model, GQA traces, and a separate non-analytic BLER map. Extensions to query-level objectives, ACK-conditioned multi-frame control, correlated packet failures, and standard-calibrated or measured packet-reliability ledgers are natural directions for future work.

\appendices

\section{Experimental Packetization and Task-Weight Instantiation}
\label{app:repro_packetization}
The optimization accepts any positive task weights fixed before link optimization. For the deterministic synthetic experiments, the token-string payload is instantiated as
\begin{equation}
b_{\rm tok}(x)=12+8\min\{|x|,12\}\quad\text{bits}.
\label{eq:app_token_length_new}
\end{equation}
Let $n_r$ be the occurrence count of relation token $r$ in the current trace window and $N_{\rm rel}=\sum_r n_r$. The normalized rarity score is
\begin{equation}
\zeta(r)=\frac{\log((N_{\rm rel}+1)/(n_r+1))}
{\max_{r'}\log((N_{\rm rel}+1)/(n_{r'}+1))+10^{-12}},
\end{equation}
and the synthetic task weight is
\begin{equation}
w_i=c_i\bigl(1+\beta_c\mathbf 1\{r_i\in\mathcal R_{\rm crit}\}\bigr)
\bigl(1+\beta_r\zeta(r_i)\bigr)\bigl(1+\beta_s s_i^{\rm safe}\bigr),
\label{eq:app_weight_new}
\end{equation}
where $s_i^{\rm safe}\in\{0,1\}$. The GQA experiment instead uses the query-derived weight specified in Table~\ref{tab:gqa_protocol}; neither choice changes the optimization structure.

\begin{table}[!t]
\caption{Semantic Packetization and Synthetic Weight Parameters}
\label{tab:appendix_payload}
\centering
\scriptsize
\setlength{\tabcolsep}{2.8pt}
\begin{tabular}{p{0.37\linewidth}p{0.56\linewidth}}
\toprule
Quantity & Value or role\\
\midrule
$B_{\rm hdr}^{v}/B_{\rm hdr}^{t}$ & 48 bits / 64 bits\\
$b_{\rm crc}^{\rm tok}/b_{\rm aux}$ & 0 bits / 10 bits\\
Vocabulary grouping & At most six new tokens per packet\\
Triplet grouping & At most six triplets per packet\\
$\mathcal R_{\rm crit}$ & crossing, near, approaching, turning, behind, overlapping\\
$\beta_c,\beta_r,\beta_s$ & $1.6,0.5,0.4$\\
Confidence & Dataset confidence when available; otherwise $c_i=1$\\
\bottomrule
\end{tabular}
\end{table}

\section{Baseline Definitions and Fair-Comparison Settings}
\label{app:baseline_details}
All policies use identical semantic frames, packet payloads, link ledgers, rate/MCS and power menus, expected-delay and energy accounting, and frame budgets. When the optional vocabulary failure cap is enabled, it is applied to every compared policy. For a fixed-weight baseline, the common packet objective is \eqref{eq:baseline_packet_objective}.

Static-DA UEP uses \eqref{eq:static_da_weights} with the full wireless action menu. Equivalently, it freezes the proposed marginal coefficients at the all-success state before any link action is selected. Dependency-agnostic HARQ gives triplet packets their aggregate task weights and allocates the same total vocabulary-weight mass across vocabulary packets in proportion to payload size, without using downstream prerequisite sets. Semantic-priority HARQ uses the triplet aggregate weights and a fixed protocol-overhead weight for every vocabulary packet. Priority UEP uses the same fixed priority rule but restricts $K_u\le K_{\rm UEP}$; UEP-only fixes $K_u=1$ while retaining unequal rate and power; Uniform HARQ exhaustively chooses one common $(m,P,K)$ tuple for all packets.

\begin{table}[!t]
\caption{Policy Definitions Used in the Experiments}
\label{tab:appendix_baselines}
\centering
\scriptsize
\setlength{\tabcolsep}{2.2pt}
\begin{tabular}{p{0.24\linewidth}p{0.39\linewidth}p{0.27\linewidth}}
\toprule
Policy & Semantic coefficient & Wireless menu\\
\midrule
Proposed & $\alpha_u(\boldsymbol\epsilon_{-u})$ from prerequisite state & Full $m,P,K$\\
Static-DA UEP & $\alpha_u(\mathbf 0)$, frozen & Full $m,P,K$\\
Dep.-agnostic & Fixed triplet weights; payload-only vocab weights & Full $m,P,K$\\
Semantic-priority & Fixed semantic ordering; fixed vocab overhead & Full $m,P,K$\\
Priority UEP & Same fixed priority & $K\le K_{\rm UEP}$\\
UEP-only & Same fixed priority & $K=1$\\
Uniform HARQ & No packet differentiation & One common tuple\\
\bottomrule
\end{tabular}
\end{table}

The three full-menu packet-adaptive policies use the same multiplier starts, step schedule, iteration cap, and feasibility-polishing rule. Static-DA is therefore not given a smaller search space than Proposed. Restricted-depth UEP baselines intentionally use the smaller HARQ sets shown in Table~\ref{tab:appendix_baselines}.

\section{Table-Driven BLER Link-Abstraction Interface}
\label{app:bler_interface}
For each action $a_u=(m_u,P_u,K_u)$, the optimizer requires only $\{\epsilon_u(a_u),\bar T_u(a_u),\bar\tau_u(a_u),\bar E_u(a_u)\}$. The retained table contains 1,152 anchors over $B\in\{112,360,760\}$ bits, $\rho\in\{0.50,0.80,1.20,1.80\}$ bpcu, $P\in\{-5,5,15,23\}$ dBm, $d\in\{80,160,240,320,410,500\}$ m, and $K\in\{1,2,3,4\}$. For payloads between anchor sizes, the implementation linearly interpolates $\log\epsilon$ in payload bits and linearly interpolates $\bar T$; payloads outside the anchor interval use the nearest payload anchor. The remaining coordinates lie directly on the table grid. Expected delay and energy are recomputed for the actual packet size using \eqref{eq:expected_delay} and \eqref{eq:expected_energy}.

Link-level PHY abstractions commonly expose decoding reliability through channel- and MCS-dependent BLER mappings rather than embedding a complete decoder in the system-level optimizer \cite{ref51}. Here, the retained ledger supplies this packet-level interface together with the corresponding HARQ statistics. A standard-specific or measured ledger can be substituted by replacing these tabulated link statistics without modifying the prerequisite model or the allocation algorithm.

\end{document}